\documentclass[runningheads]{llncs}
\usepackage{xspace}
\usepackage{hyperref}
\usepackage{amssymb,amsmath}

\newcommand{\OPT}{\textsc{Opt}\xspace}
\newcommand{\ALG}{\textsc{Alg}\xspace}
\newcommand{\REORG}{\textsc{Comp}\xspace}
\newcommand{\calG}{{\cal G}}

\newcommand{\calR}{{\cal R}}

\newcommand{\calL}{{\cal L}}
\newcommand{\calC}{{\cal C}}
\newcommand{\vecc}{\mathbf{c}}
\newcommand{\vecx}{\mathbf{x}}
\newcommand{\vecy}{\mathbf{y}}

\newcommand{\vecz}{\mathbf{z}}
\newcommand{\vecw}{\mathbf{w}}
\newcommand{\vecg}{\mathbf{g}}
\newcommand{\vech}{\mathbf{h}}
\newcommand{\vecu}{\mathbf{u}}
\newcommand{\veczero}{\mathbf{0}}
\newcommand{\vecchat}{\mathbf{\hat{c}}}
\newcommand{\configs}{\calC}
\newcommand{\configsext}{\calC^\mathrm{ext}}

\newcommand{\Z}{\mathbb{Z}}
\newcommand{\e}{\mathrm{e}}
\newcommand{\elem}{\mathrm{nd}}

\begin{document}

\title{Improved Analysis of Online Balanced Clustering\thanks{Supported 
by Polish National Science Centre grant 2016/22/E/ST6/00499, 
by Center for Foundations of Modern Computer Science (Charles University project UNCE/SCI/004), 
by the project 19-27871X of GA ČR, and by NSF CAREER grant 1651861.}}

\author{Marcin Bienkowski\inst{1} \and
Martin Böhm\inst{1} \and
Martin Kouteck\'{y}\inst{2} \and
Thomas Rothvoß\inst{3} \and
Ji\v{r}\'{\i} Sgall\inst{2} \and
Pavel Vesel\'{y}\inst{2}}

\authorrunning{M.~Bienkowski, M.~Böhm, M.~Kouteck\'{y}, T.~Rothvoß, J.~Sgall, P.~Vesel\'{y}}

\institute{Institute of Computer Science, University of Wroc{\l}aw, Poland \and
Computer Science Institute of Charles University, Faculty of Mathematics and Physics, Prague, Czechia \and
University of Washington, Seattle, United States}

\maketitle

\begin{abstract}
In the online balanced graph repartitioning problem, one has to maintain a
clustering of $n$ nodes into $\ell$ clusters, each having $k = n / \ell$ nodes.
During runtime, an online algorithm is given a stream of communication requests
between pairs of nodes: an inter-cluster communication costs one unit, while the
intra-cluster communication is free. An algorithm can change the clustering,
paying unit cost for each moved node.

This natural problem admits a simple $O(\ell^2 \cdot k^2)$-competitive algorithm
\REORG, whose performance is far apart from the best known lower bound of
$\Omega(\ell \cdot k)$. One of open questions is whether the dependency on
$\ell$ can be made linear; this question is of practical importance as in the
typical datacenter application where virtual machines are clustered on physical
servers, $\ell$ is of several orders of magnitude larger than~$k$. We answer
this question affirmatively, proving that a simple modification of \REORG is
$(\ell \cdot 2^{O(k)})$-competitive.

On the technical level, we achieve our bound by translating the problem to a system
of linear integer equations and using Graver bases to show the existence of a
``small'' solution.

\end{abstract}


\section{Introduction}

We study the \emph{online balanced graph repartitioning problem}, introduced by
Avin et al.~\cite{AvLoPS16}. In this problem, an algorithm has to maintain a
time-varying partition of $n$ nodes into $\ell$ clusters, each having $k = n /
\ell$ nodes. An algorithm is given an online stream of communication requests,
each involving a pair of nodes. A~communication between a pair of nodes from the
same cluster is free, while inter-cluster communication incurs unit cost. In
response, an~algorithm may change the mapping of nodes to clusters, also paying
a unit cost for changing a~cluster of a single node. After remapping, each cluster
has to contain $k$ nodes again. 

We focus on an online scenario, where an (online) algorithm has to make
irrevocable remapping decisions after each communication request without the
knowledge of the future. The problem can be seen as a dynamic (and online)
counterpart of a so-called $\ell$-balanced graph partitioning
problem~\cite{AndRae06}, where the goal is to partition the graph into $\ell$
equal-size parts to minimize the total weight of edges in the cut. In
particular, $\ell = 2$ corresponds to the well-studied graph  bisection
problem~\cite{KraFei06,Raec08}. 

A main practical motivation originates from server virtualization in
datacenters. There, nodes correspond to $n$ virtual machines run on $\ell$
physical ones. Each physical machine (a~cluster) has the capacity for
accommodating $k$ virtual machines. Communication requests are induced by
distributed applications running in the datacenter. While communication within a
physical machine is practically free, the inter-cluster communication (between
different physical machines) generates considerable load and affects the overall
running time (see, e.g.,~\cite{ChZMJS11}). Due to the current capabilities of modern
architectures, migrating a virtual machine to another physical machine
(remapping of a node) is possible, but it incurs a~certain load.\footnote{In
the reality, the network load generated by a single request is much smaller than
the load of migrating a~whole virtual machine. This has been captured by some papers, 
which assigned cost $\alpha \geq 1$ to the latter event ($\alpha$ is
then a parameter of the problem). However, this additional difficulty can be
resolved by standard rent-or-buy approaches (reacting only to every $\alpha$-th
request between a given pair of nodes). Therefore, and also to keep the description
simple, in this paper, we assume that $\alpha = 1$.}

To evaluate the performance of an online algorithm \ALG, we use 
a standard notion of competitive ratio~\cite{BorEl-98} which is the 
supremum over all possible inputs of \ALG-to-\OPT cost, where \OPT denotes the optimum offline
solution for the problem.

\subsection{Component-Based Approach}
\label{sec:natural}

Our contribution is related to the following natural algorithm (henceforth
called \REORG) proposed by Avin et al.~\cite{AvLoPS16}. We define it in detail
in \autoref{sec:alg_def}; here, we give its informal description. \REORG operates
in phases, in each phase keeping track of components of nodes that communicated
in this phase (i.e., connected components of a graph whose edges are communication
requests). \REORG always keeps all nodes of a given component in a single
cluster; we call this property \emph{component invariant}. When components are
modified, to maintain the invariant, some nodes may have to change their clusters; 
the associated cost can be trivially bounded by $n$.
As this changes the mapping of nodes to clusters, we call it \emph{remapping event}.
If such remapping does not exist\footnote{Deciding whether such remapping exists
is NP-hard. As typical for online algorithms, however, our focus is on studying
the disadvantage of not knowing the future rather than on computational
complexity.}, the current phase terminates, and it is possible to show that \OPT
paid at least $1$ in this phase. The number of remapping events is equal to the
number of times connected components are modified, which can be upper-bounded by
$n-1$. Thus, the overall cost of \REORG in a~single phase is $O(n^2) = O(\ell^2
\cdot k^2)$, while that of \OPT is at least $1$. This shows that \REORG is
$O(\ell^2 \cdot k^2)$-competitive.

\subsection{Related Work}

Perhaps surprisingly, no better algorithm than \REORG (even randomized one) is
known for the general case. Some improvements were, however, given for specific
values of $k$ and $\ell$. 

The dependency of the competitive ratio on $k$ is at least linear for
deterministic algorithms: a lower bound of $\Omega(k)$ follows by the reduction
from online paging~\cite{SleTar85} and holds already for $\ell = 2$
clusters~\cite{AvBLPS20}. For $k = 2$, $O(1)$-competitive (deterministic)
algorithms are known: a 7-competitive algorithm was given by Avin et
al.~\cite{AvLoPS16,AvBLPS20} and was later improved to a~6-competitive one by
Pacut et al.~\cite{PaPaSc21}. However, already for $k \geq 3$, the competitive
ratio of any deterministic algorithm cannot be better than $\Omega(k \cdot
\ell)$~\cite{PaPaSc20,PaPaSc21}. For the special case of $k = 3$, Pacut et
al.~\cite{PaPaSc21} showed that a variant of \REORG is $O(\ell)$-competitive.
The lack of progress towards improving the $O(\ell^2 \cdot k^2)$ upper bound for
the general case motivated the research of simplified variants. 

Henzinger et
al.~\cite{HeNeSc19} initiated the study of a~so-called \emph{learning variant},
where there exists a fixed partitioning (unknown to an algorithm) of nodes into
clusters, and the communication requests are \emph{consistent} with this mapping,
i.e., all requests are given between same-cluster node pairs. Hence, the implicit goal 
of an~algorithm is to learn such static mapping. The deterministic lower
bound of $\Omega(k \cdot \ell)$ also holds for this variant, and, furthermore, there
exists a deterministic algorithm that asymptotically matches this
bound~\cite{PaPaSc20,PaPaSc21}. 

Another strand of research focused on a resource-augmented scenario, where 
each cluster of an online algorithm is able to accommodate $(1 + \epsilon) \cdot k$
nodes (but the online algorithm is still compared to \OPT, whose clusters have
to keep $k$ nodes each). For $\epsilon > 1$, the first deterministic algorithm
was given by Avin et al.~\cite{AvBLPS20}; the achieved ratio was $O(k \cdot \log
k)$ (for an arbitrary $\ell$). Suprisingly, the ratio remains $\Omega(k)$ even 
for large $\epsilon$ (as long as the algorithm cannot keep all nodes in a single 
cluster)~\cite{AvBLPS20}.

When these two simplifications are combined (i.e., the learning variant is
studied in a resource-augmented scenario), asymptotically optimal results are 
due to Henzinger et al.~\cite{HeNeSc19,HeNeRS21}. They show that for any fixed
$\epsilon > 0$, the deterministic ratio is $\Theta(\ell \cdot \log
k)$~\cite{HeNeSc19,HeNeRS21} and the randomized ratio is $\Theta(\log \ell + \log
k)$~\cite{HeNeRS21}. Furthermore, for $\epsilon > 1$, their deterministic
algorithm is $O(\log k)$-competitive~\cite{HeNeRS21}.

\subsection{Our Contribution}

We focus on the general variant of the balanced graph repartitioning problem. We
study a variant of \REORG (see \autoref{sec:natural})
in which \emph{each remapping
event is handled in a way minimizing the number of affected clusters}. 

We show
that the number of nodes remapped this way is $2^{O(k)}$ (in comparison to the
trivial bound of $n = \ell \cdot k$). The resulting bound on the competitive
ratio is then $(\ell \cdot k) \cdot 2^{O(k)} = \ell \cdot 2^{O(k)}$, i.e., we
replaced the quadratic dependency on $\ell$ in the competitive ratio by the
linear one. We note that the resulting algorithm retains the $O(\ell^2 \cdot
k^2)$-competitiveness guarantee of the original \REORG algorithm as well.

Given the lower bound of $\Omega(\ell \cdot k)$~\cite{PaPaSc20,PaPaSc21}, the
resulting strategy is optimal for a constant~$k$. We also note that for the
datacenter application described earlier, $k$ is of several orders of magnitude
smaller than $\ell$.

We achieve our bound by translating the remapping event to a system of
linear integer equations so that the size of the solution (sum of
values of variables) is directly related to the number of affected
clusters. Then, we use algebraic tools such as Graver bases to argue
that these equations admit a ``small'' solution.\footnote{
One could also bound the size of a solution along the lines of 
Schrijver~\cite[Corollary 17.1b]{Schrij98}, which boils down to a determinant bound, same as our proof. In
general, Graver basis elements may be much smaller, but in out specific case,
the resulting bound would be asymptotically the same.}


\section{Preliminaries}

An offline part of the input is a set of $n = \ell \cdot k$ nodes and
an initial valid mapping of these nodes into $\ell$ clusters. We call
a mapping \emph{valid} if each cluster contains exactly $k$ nodes.

An online part of the input is a set of requests, each being a pair of nodes
$(u,v)$. The request incurs cost $1$ if $u$ and $v$ are mapped to different
clusters. After each request, an~online algorithm may modify the current node
mapping to a~new valid one, paying $1$ for each node that changes its cluster.

For an input $I$ and an algorithm $\ALG$, $\ALG(I)$ denotes its cost
on input~$I$, whereas $\OPT(I)$ denotes the optimal cost of an offline
solution. \ALG is $\gamma$-competitive if there exists a~constant
$\beta$ such that for any input~$I$, it holds that $\ALG(I) \leq
\gamma \cdot \OPT(I) + \beta$. While $\beta$ has to be independent 
of the online part of the sequence, it may depend on the offline 
part, i.e., be a function of parameters~$\ell$ and $k$.


\section{Better Analysis for COMP}
\label{sec:alg_def}

Algorithm \REORG~\cite{AvLoPS16} splits input into phases.  
In each phase, \REORG maintains an auxiliary partition $\calR$ of the set of
nodes into \emph{components}; initially, each node is in its own
singleton component. \REORG maintains the following \emph{component invariant}: 
for each component $S \in \calR$, all its nodes are inside the same cluster.

Assume a request $(u,v)$ arrives. Two cases are possible. 

\begin{itemize}
\item If $u$ and $v$ are within the same component of $\calR$, then by the
component invariant, they are in the same cluster. \REORG serves this
request without paying anything, without changing $\calR$, and without remapping nodes.
\item 
If $u$ and $v$ are in different components $S_a$ and $S_b$, \REORG
merges these components into $S_{ab} = S_a \uplus S_b$ by removing $S_a$ and $S_b$ from $\calR$, and
adding $S_{ab}$ to $\calR$. (By $A \uplus B$ we denote the \emph{disjoint union} of sets $A$
and $B$.)

If components $S_a$ and $S_b$ were in different clusters, the resulting
component~$S_{ab}$ now spans two clusters, which violates the component
invariant. To restore it, \REORG verifies whether there exists a valid mapping
of nodes into clusters preserving the component invariant (also for the new
component $S_{ab}$). In such case, a \emph{remapping event} occurs: among all
such mappings, \REORG chooses one minimizing the total number of affected
clusters. (The original variant of \REORG~\cite{AvLoPS16} simply chose any such mapping.) If,
however, no such valid mapping exists, \REORG resets $\calR$ to the initial
partition in which each node is in its own component and starts a new phase. 
\end{itemize}

\begin{lemma}
\label{lem:remapping_to_ratio}
Assume each remapping event affects at most $f(\ell, k)$ clusters (for some
function~$f$). Then \REORG is $O(\ell \cdot k^2 \cdot f(\ell, k))$-competitive.
\end{lemma}

\begin{proof}
Fix any input $I$ and split it into phases according to \REORG. Each phase (except possibly the last
one) terminates because there is no valid node mapping that would preserve the
component invariant. That is, for any fixed valid mapping of nodes to clusters,
the phase contains an~inter-cluster request. Thus, during any phase, either $\OPT$ pays
at least $1$ for node remapping, or its mapping is fixed within phase, and then it pays
at least $1$ for serving requests.

On the other hand, within each phase, \REORG modifies family $\calR$ of
components at most $n-1 = \ell \cdot k - 1$ times since the number of components
decreases by one with each modification. Each time it happens, it pays
$1$ for the request, and then, if the remapping event is triggered, it pays
additionally at most $k \cdot f(\ell, k)$ as it remaps at most $k$ nodes from each affected cluster.
Hence, the overall cost in a single phase is $(\ell \cdot k - 1)
\cdot (1 + k \cdot f(\ell, k)) = O(\ell \cdot k^2 \cdot f(\ell, k))$. 

The cost of \REORG in the last phase is universally bounded by $O(\ell \cdot k^2 \cdot f(\ell, k))$ 
and in the remaining phases, the \REORG-to-\OPT cost ratio is at most $O(\ell \cdot k^2 \cdot
f(\ell, k))$, which concludes the proof.
\qed
\end{proof}

The trivial upper bound on $f(\ell, k)$ is $\ell$, which together with
\autoref{lem:remapping_to_ratio} yields the already known bound of $O(\ell^2 \cdot
k^2)$~\cite{AvLoPS16}. In the remaining part, we show that if \REORG{} tries to
minimize the number of affected clusters for each remapping event, then $f(\ell,
k)$ can be upper-bounded by $2^{O(k)}$. Thus, our analysis beats the simple
approach when the number of clusters~$\ell$ is much larger than the cluster
capacity $k$. This ratio is also optimal for constant $k$ as the lower bound on
the competitive ratio is $\Omega(\ell \cdot k)$~\cite{PaPaSc20,PaPaSc21}.

\subsection{Analyzing a Remapping Event}

Recall that we want to analyze a remapping event, where we have a given valid
mapping of nodes to clusters and a family of components $\calR$ satisfying the
component invariant. \REORG merges two components $S_a$ and $S_b$ into one, and the
resulting component $S_{ab} = S_a \uplus S_b$ spans two clusters. As we assume
that it is possible to remap nodes to satisfy the component invariant, the size
of $S_{ab}$ is at most $k$.

We first express our setup in the form of a system of linear integer equations.
The definition below assumes that component invariant holds, i.e., each
component is entirely contained in some cluster.

\begin{definition}[Cluster configuration]
A \emph{configuration} of a cluster $C$ is a vector
$\vecc = (c_1, \dots, c_k)$ where $c_i \geq 0$ is the number of
components of size $i$ in $C$.
\end{definition}

We denote the set of all possible cluster configurations by $\calC$, i.e.,
$\calC$ contains all possible $k$-dimensional vectors $\vecc = (c_1, \dots,
c_k)$ such that $\sum_{i=1}^k i \cdot c_i = k$. 
For succinctness, we use $\elem(\vecc) = \sum_{i=1}^k i \cdot c_i$.
The number of
configurations is equal to the partition function $\pi(k)$, which
denotes the number of possibilities how $k$ can be expressed as a sum
of a multiset of non-negative integers. From the known bounds on $\pi$, we
get $|\calC|=\pi(k) \leq 2^{O(\sqrt{k})}$~\cite{Erdos42}.

We take two clusters that contained components $S_a$ and $S_b$, respectively,
and we virtually replace them by a~\emph{pseudo-cluster} $\hat{C}$ that contains
all their components (including $S_{ab}$ and excluding $S_a$ and $S_b$). Let $\vecchat$ be the
configuration of this pseudo-cluster; note that 
$\elem(\vecchat) = 2 k$, and hence $\vecchat \notin \configs$. We define the extended set of
configurations $\configsext = \configs \uplus \{ \vecchat \}$.

Let $\vecx$ be a $|\configsext|$-dimensional vector, indexed by possible configurations
from $\configsext$, such that, for any configuration $\vecc \in \configsext$,
$x_\mathbf{c}$ is the number of clusters with configuration~$\mathbf{c}$    
before remapping event. Note that $\vecx \geq \veczero$, 
$x_\vecchat = 1$ and $\|\vecx\|_1 = \sum_{\vecc \in \configsext} x_\vecc = \ell - 1$.
Let 
\begin{equation}
\label{eq:u_def}
    \vecu = \sum_{\vecc \in \configsext} x_\vecc \cdot \vecc\,.
\end{equation}
That is, $\vecu = (u_1, u_2, \ldots, u_k)$, where $u_i$ is the total number of components 
of size~$i$ (in all clusters). Clearly, $\elem(\vecu) = (\ell-2) \cdot k + 2 \cdot k = \ell \cdot k$.
We rewrite~\eqref{eq:u_def} as 
\begin{equation}
\label{eq:u_def_matrix}
    \vecu = A \vecx\,,
\end{equation}
where $A$ is the matrix with $k$ rows and $|\configsext|$ columns. Its columns are equal to
vectors of configurations from $\configsext$.

As $\vecx$ describes the current state of the clusters, in the following, we focus on
finding an appropriate vector $\vecy$ describing a \emph{target} state of the clusters,
i.e., their state after the remapping event takes place. 

\begin{definition}
An integer vector $\vecy$ is a \emph{valid target vector} if it is 
$|\configsext|$-di\-men\-sion\-al and satisfies $\vecy \geq 0$, 
$A \vecy = \vecu$ and $y_\vecchat = 0$.
\end{definition}

\begin{lemma}
\label{obs:valid_target_ell}
For any valid target vector $\vecy$, it holds that $\|\vecy\|_1 = \ell$.
\end{lemma}

\begin{proof}
Let $\ell' = \|\vecy\|_1$. Then $\vecy = \sum_{i=1}^{\ell'} 
\vecy^i$, where $\vecy^i$ is equal to $1$ for some configuration $\vecc \neq \vecchat$ 
and $0$ everywhere else. 
As $\vecu = A \vecy = \sum_{i=1}^{\ell'} A \vecy^i$, we obtain
$\ell \cdot k = \elem(\vecu) = \sum_{i=1}^{\ell'} \elem(A \vecy^i)$. 
For any $i$, vector $A \vecy^i$ is a single column of $A$ corresponding to 
a configuration $\vecc \neq \vecchat$, and thus $\elem(A \vecy^i) = k$.
This implies that $\ell' = \ell$, which concludes the proof.
\qed
\end{proof}

\begin{lemma}
\label{obs:solution_exists}
There exist a valid target vector $\vecy$.
\end{lemma}

\begin{proof}
After the remapping event takes place, it is possible to map nodes to different
clusters so that the component invariant holds (after merging $S_a$ and $S_b$
into $S_{ab}$), i.e., each component is entirely contained in some cluster (not
in the pseudo-cluster). Thus, each cluster has a~well-defined configuration in
$\configs$, and $y_\vecc$ is simply the number of clusters with
configuration~$\vecc$ after remapping. 
\qed
\end{proof}

\begin{lemma}
\label{lem:solution_cost}
Fix a valid target vector $\vecy$.
Then, there exists a node remapping that affects 
$(1/2) \cdot \|\vecx - \vecy \|_1 + 1/2$ clusters.
\end{lemma}

\begin{proof}
We define vector $\tilde{\vecx}$, such that $\tilde{x}_\vecc = x_\vecc$ for
$\vecc \neq \vecchat$, and $\tilde{x}_\vecchat = 2$. Hence $\|\tilde{\vecx}\|_1
= \ell$. By \autoref{obs:valid_target_ell}, $\|\vecy\|_1 = \ell$ as well.
For $\vecc \in \configs$, the value of $\tilde{x}_\vecc$
denotes how many clusters have configuration $\vecc$, with $\tilde{x}_\vecchat =
2$ simply denoting that there are two clusters whose configuration is not equal to
any configuration from $\configs$.

Now, for any configuration $\vecc \in \configsext$, we fix $\min\{\tilde{x}_\vecc,
y_\vecc\}$ clusters with configuration $\vecc$. Our remapping does not touch
these clusters and there exists a straightforward node remapping which involves
only the remaining clusters. Their number is equal to 
\begin{align*}
(1/2) \cdot \sum_{\vecc \in \configsext} |\tilde{x}_\vecc - y_\vecc| 
& = 1 + (1/2) \cdot \sum_{\vecc \in \configs} |\tilde{x}_\vecc - y_\vecc| 
= 1 + (1/2) \cdot \sum_{\vecc \in \configs} |x_\vecc - y_\vecc| \\
& = (1/2) + (1/2) \cdot \sum_{\vecc \in \configsext} |x_\vecc - y_\vecc| \\
& = (1/2) + (1/2) \cdot \|\vecx - \vecy \|_1\,,
\end{align*}
and thus the lemma follows.
\qed
\end{proof}

We note that a valid target vector $\vecy$ guaranteed by \autoref{obs:solution_exists} may be 
completely different from vector $\vecx$ describing the current clustering,
and thus it is possible that $\|\vecx - \vecy \|_1 = \Omega(\ell)$. We however show that 
on the basis of $\vecy$, we may find a valid target vector $\vecy'$,
such that $\|\vecx - \vecy' \|_1$ is small, i.e., at most $2^{O(k)}$.

\subsection{Using Graver Basis}

\autoref{obs:solution_exists} guarantees the existence of vector $\vecz = \vecx
- \vecy$, encoding the reorganization of the clusters. We already know that $A
\vecz = \veczero$ must hold, but there are other useful properties as well; for
instance, if $z_\vecc > 0$ for a configuration $\vecc$, then $z_\vecc \leq
x_\vecc$ (i.e. the corresponding reorganization does not try to remove more
clusters of configuration $\vecc$ than $x_\vecc$). Our goal is to find another
vector $\vecw$ that also encodes the reorganization and $\|\vecw \|_1$ is small.

The necessary condition for $\vecw$ is that it satisfies $A \vecw = \veczero$, and thus we 
study properties of matrix $A$, defined by \eqref{eq:u_def_matrix}, 
in particular its Graver basis. For an introduction to Graver bases,
we refer the interested reader to a book by Onn~\cite{Onn10}.

\begin{definition}[Sign-compatibility]
Given two vectors $\mathbf{a}$ and $\mathbf{b}$ of the same length, we say that
they are \emph{sign-compatible} if for each coordinate $i$ the sign of~$a_i$ is
the same as the sign of $b_i$ (i.e., $a_i \cdot b_i \geq 0$). We say that
$\mathbf{a}\sqsubseteq \mathbf{b}$ if $\mathbf{a}$ and~$\mathbf{b}$ are
sign-compatible and $|a_i|\leq |b_i|$ for every coordinate $i$. Note that
$\sqsubseteq$ imposes a~partial order.
\end{definition}

\begin{definition}[Graver basis]
Given an integer matrix $A$, its \emph{Graver basis} $\calG(A)$ is the set of
$\sqsubseteq$-minimal elements of the \emph{lattice} $\calL^*(A) = \{\vech
\,|\, A \vech = \mathbf{0}, \vech \in \Z^n, \vech \neq \veczero \}$. 
\end{definition}

\begin{lemma}[Lemma 3.2 of \cite{Onn10}]
\label{lem:graver_decomposition}
Any vector $\mathbf{h} \in \calL^*(A)$ is a sign-compatible sum $\mathbf{h} =
\sum_i \mathbf{g}^i$ of Graver basis elements $\mathbf{g}^i\in \calG(A)$, with
some elements possibly appearing with repetitions.
\end{lemma}

Having the tools above, we may now prove the existence of a remapping involving 
small number of clusters.

\begin{lemma}
\label{lem:better_solution}
If there exists a valid target vector $\vecy$, then 
there exists a valid target vector~$\vecy'$, such that $\vecx - \vecy' \in \calG(A)$.
\end{lemma}

\begin{proof}
By \eqref{eq:u_def_matrix} and the lemma assumption, $A \vecx = A \vecy$. 
Let $\vecz = \vecx - \vecy$. Then, $z_\vecchat = x_\vecchat - y_\vecchat = 1$ and 
$A \vecz = \veczero$, and thus $\vecz \in \calL^*(A)$.

Using \autoref{lem:graver_decomposition}, we may express $\vecz$ as 
$\vecz = \sum_i \mathbf{g}^i$, where
$\vecg^i \in \calG(A) \subseteq \calL^*(A)$ for all $i$
and all $\vecg^i$ are sign-compatible with $\vecz$. 
As $z_\vecchat = 1$, the sign-compatibility means that 
there exists $\vecg^j$ appearing in the sum with $g^j_\vecchat = 1$.
We set $\vecw = \vecg^j$.

Let $\vecy' = \vecx - \vecw$. Clearly 
$\vecx - \vecy' = \vecw \in \calG(A)$. It remains to show that 
$\vecy'$ is a~valid target vector.
We have $A \vecy' = A \vecx - A \vecw = \vecu - \veczero = \vecu$ and 
$y'_\vecchat = x_\vecchat - w_\vecchat = 1 - 1 = 0$. 
To show that $\vecy' \geq \veczero$, we consider two cases. If $z_\vecc \geq 0$, then by sign-compatibility 
$0 \leq w_\vecc \leq z_\vecc$, and thus $y'_\vecc = x_\vecc - w_\vecc
\geq x_\vecc - z_\vecc = y_\vecc \geq 0$. On the other hand, if $z_\vecc < 0$,
then again by sign-compatibility, $w_\vecc \leq 0$, and thus $y'_\vecc =
x_\vecc - w_\vecc \geq x_\vecc \geq 0$.
\qed
\end{proof}

To complete our argument, it remains to bound $\|\vecg\|_1$, where 
$\vecg$ is an arbitrary element of~$\calG(A)$. 
We start with a known bound for $\ell^\infty$-norm of any element of the Graver basis.

\begin{lemma}[Lemma 3.20 of~\cite{Onn10}]
\label{clm:onn-Lemma3.20}
Let $q$ be the number of columns of integer matrix $M$. Let $\Delta(M)$ denote
the maximum absolute value of the determinant of a square sub-matrix of $M$. Then 
$\| \mathbf{g} \|_\infty \le q \cdot \Delta(M)$ for any $\mathbf{g}\in \calG(M)$.
\end{lemma}

\begin{lemma}
\label{lem:graver_norm_bound}
For any $\vecg \in \calG(A)$, it holds that $\|\vecg\|_1 \leq 2^{O(k)}$.
\end{lemma}

\begin{proof}
We start by showing that $\Delta(A) \leq \e^k$. Let $\tilde{A}$ be the matrix
$A$ with each row multiplied by its index. By the definition of a configuration,
the column sums of $\tilde{A}$ are equal to $k$, with the exception of the
column corresponding to the configuration~$\vecchat$ whose sum is equal to $2k$.
As all entries of $\tilde{A}$ are non-negative, the same holds for
$\ell_1$-norms of columns of $\tilde{A}$.

Fix any square sub-matrix $B$ of $A$ and let $j$ be the number of its columns
(rows). Let $\tilde{B}$ be the corresponding sub-matrix of $\tilde{A}$. Since
$\tilde{B}$ is obtained from $B$ by multiplying its rows by $j$ distinct
positive integers, $|\det(\tilde{B})|\geq j!\cdot|\det(B)|$. (This relation
holds with equality if and only if $B$ contains (a part of)
the first $j$ rows of $A$.) 

It therefore remains to upper-bound $|\det(\tilde{B})|$. Hadamard's bound on
determinant states that the absolute value of a determinant is at most the
product of lengths ($\ell_2$-norms) of its column vectors, which are in turn
bounded by the $\ell_1$-norms of columns of $\tilde{B}$. These are not greater
than $\ell_1$-norms of the corresponding columns of $\tilde{A}$ and thus
$|\det(\tilde{B})| \leq 2 \cdot k^j$.

Combining the above bounds and using $j \leq k$, we obtain
\[
    |\det(B)|\leq\frac{|\det(\tilde{B})|}{j!}
    \leq \frac{2\cdot k^j}{j!}\leq \frac{2\cdot k^k}{k!}\leq 2 \cdot \e^{k-1}\leq \e^k\,.
\]

As $B$ was chosen as an arbitrary square sub-matrix of $A$, $\Delta(A) \leq \mathrm{e}^k$.
Our matrix $A$ has $|\configsext|$ columns, and hence
\autoref{clm:onn-Lemma3.20} implies that $\| \vecw \|_\infty\le |\configsext| \cdot 2^{O(k)}$.
As $\vecw$ is $|\configsext|$-dimensional, 
$\| \vecw \|_1 \leq |\configsext| \cdot \| \vecw \|_\infty 
\leq |\configsext|^2 \cdot \e^k$. Finally using
$|\configsext| \leq 2^{O(\sqrt{k})}$, we obtain 
$\| \vecw \|_1 \leq 2^{O(k)}$.
\qed
\end{proof}

\begin{corollary}
\label{cor:final}
The remapping event of \REORG affects at most $2^{O(k)}$ clusters.
\end{corollary}

\begin{proof}
Combining \autoref{obs:solution_exists} with \autoref{lem:better_solution}
yields the existence of a valid target vector~$\vecy'$ satisfying $\|\vecx -
\vecy'\|_1 \in \calG(A)$. By \autoref{lem:graver_norm_bound}, $\|\vecx -
\vecy'\|_1 \leq 2^{O(k)}$. Thus, plugging $\vecy'$ to
\autoref{lem:solution_cost} yields the corollary.
\qed
\end{proof}

\subsection{Competitive Ratio}

Combining \autoref{lem:remapping_to_ratio} with \autoref{cor:final} immediately 
yields the desired bound on the competitive ratio of \REORG.

\begin{theorem}
The variant of \REORG in which 
each remapping event is handled 
in a way minimizing the number of affected clusters 
is $(\ell \cdot 2^{O(k)})$-competitive.
\end{theorem}

\bibliographystyle{splncs04}
\bibliography{references}

\begin{thebibliography}{10}
\providecommand{\url}[1]{\texttt{#1}}
\providecommand{\urlprefix}{URL }
\providecommand{\doi}[1]{https://doi.org/#1}

\bibitem{AndRae06}
Andreev, K., R{\"a}cke, H.: Balanced graph partitioning. Theory of Computing
  Systems  \textbf{39}(6),  929--939 (2006). \doi{10.1007/s00224-006-1350-7}

\bibitem{AvBLPS20}
Avin, C., Bienkowski, M., Loukas, A., Pacut, M., Schmid, S.: Dynamic balanced
  graph partitioning. SIAM Journal on Discrete Mathematics  \textbf{34}(3),
  1791--1812 (2020). \doi{10.1137/17M1158513}

\bibitem{AvLoPS16}
Avin, C., Loukas, A., Pacut, M., Schmid, S.: Online balanced repartitioning.
  In: Proc. 30th Int. Symp. on Distributed Computing (DISC). pp. 243--256
  (2016). \doi{10.1007/978-3-662-53426-7\_18}

\bibitem{BorEl-98}
Borodin, A., {El-Yaniv}, R.: Online Computation and Competitive Analysis.
  Cambridge University Press (1998)

\bibitem{ChZMJS11}
Chowdhury, M., Zaharia, M., Ma, J., Jordan, M.I., Stoica, I.: Managing data
  transfers in computer clusters with {Orchestra}. In: ACM SIGCOMM. pp. 98--109
  (2011). \doi{10.1145/2018436.2018448}

\bibitem{Erdos42}
Erdős, P.: On an elementary proof of some asymptotic formulas in the theory of
  partitions. Annals of Mathematics  \textbf{43}(3),  437--450 (1942).
  \doi{10.2307/1968802}

\bibitem{HeNeRS21}
Henzinger, M., Neumann, S., R{\"{a}}cke, H., Schmid, S.: Tight bounds for
  online graph partitioning. In: Proc. 32nd ACM-SIAM Symp. on Discrete
  Algorithms (SODA). pp. 2799--2818 (2021). \doi{10.1137/1.9781611976465.166}

\bibitem{HeNeSc19}
Henzinger, M., Neumann, S., Schmid, S.: Efficient distributed workload
  (re-)embedding. In: 2019 SIGMETRICS/Performance Joint International
  Conference on Measurement and Modeling of Computer Systems. pp. 43--44
  (2019). \doi{10.1145/3309697.3331503}

\bibitem{KraFei06}
Krauthgamer, R., Feige, U.: A polylogarithmic approximation of the minimum
  bisection. SIAM Review  \textbf{48}(1),  99--130 (2006).
  \doi{10.1137/050640904}

\bibitem{Onn10}
Onn, S.: Nonlinear discrete optimization. Zurich Lectures in Advanced
  Mathematics, European Mathematical Society (2010)

\bibitem{PaPaSc20}
Pacut, M., Parham, M., Schmid, S.: Brief announcement: Deterministic lower
  bound for dynamic balanced graph partitioning. In: Proc. 39th ACM Symp. on
  Principles of Distributed Computing (PODC). pp. 461--463 (2020).
  \doi{10.1145/3382734.3405696}

\bibitem{PaPaSc21}
Pacut, M., Parham, M., Schmid, S.: Optimal online balanced graph partitioning.
  In: Proc. 40th IEEE Int. Conf. on Computer Communications (INFOCOM). pp.~1--9
  (2021). \doi{10.1109/INFOCOM42981.2021.9488824}

\bibitem{Raec08}
R{\"{a}}cke, H.: Optimal hierarchical decompositions for congestion
  minimization in networks. In: Proc. 40th ACM Symp. on Theory of Computing
  (STOC). pp. 255--264 (2008). \doi{10.1145/1374376.1374415}

\bibitem{Schrij98}
Schrijver, A.: Theory of linear and integer programming. John Wiley \& Sons
  (1998)

\bibitem{SleTar85}
Sleator, D.D., Tarjan, R.E.: Amortized efficiency of list update and paging
  rules. Communications of the ACM  \textbf{28}(2),  202--208 (1985).
  \doi{10.1145/2786.2793}

\end{thebibliography}

\end{document}